%% file: main.tex
\documentclass{article}

\input{preamble}

\title{\name: Decentralized Storage Designed to Serve}
\author[1]{Guy Goren}
\author[1]{Andrew Hariri}
\author[2]{Timothy D. R. Hartley}
\author[2]{Ravi Kappiyoor}
\author[1]{Alexander Spiegelman}
\author[2]{David Zmick}

\affil[1]{Aptos Labs}
\affil[2]{Jump Crypto}
\date{}

\begin{document}
\maketitle

\begin{abstract}
    Existing decentralized storage protocols fall short of the service
    required by real-world applications. Their throughput, latency, cost-effectiveness, and availability are insufficient for demanding workloads such as video streaming, large-scale data analytics, or AI training. As a result, Web3 data-intensive applications are predominantly dependent on centralized infrastructure.
    
    \name{} is a high-performance decentralized storage protocol designed to meet demanding needs. It achieves fast, reliable access to large volumes of data while preserving decentralization guarantees. The architecture reflects lessons from Web2 systems: it separates control and data planes, uses erasure coding with low replication overhead and minimal repair bandwidth, and operates over a dedicated backbone connecting RPC and storage nodes. Reads are paid, which incentivizes good performance. \name{} also introduces a novel auditing protocol that provides strong cryptoeconomic guarantees without compromising performance, a common limitation of other decentralized solutions. The result is a decentralized system that brings Web2-grade performance to production-scale, read-intensive Web3 applications.%
    \footnote{Disclaimer: For convenience and consistency, this paper describes Shelby in the present tense.  As a result, this paper includes certain forward-looking statements based on the Shelby design and expectations of the authors as of the date of this paper, but please note that future facts or operation of the Shelby protocol may differ from the content expressed herein.}
\end{abstract}

\section{Introduction}
\label{sec:Introduction}

Web3 has made major strides in decentralized payments and trading use cases.  Despite this progress, the ecosystem still lacks the core infrastructure needed to scale: decentralized cloud services. Without high-performance alternatives to cloud storage, networking, and compute, Web3 applications hit hard scalability limits, constraining both what they can do and who they can serve.

A critical missing piece is \textit{hot storage}: infrastructure that can support low-latency, high-throughput reads. Use cases like video streaming, real-time collaboration, and AI inference all rely on fast, reliable access to large datasets. While protocols like Filecoin~\cite{filecoin}, Arweave~\cite{arweave}, Walrus~\cite{walruswhitepaper} and Celestia~\cite{celestiareplication} support archival or cold storage, none can meet the performance requirements of dynamic, read-heavy applications.

The barrier is not just technical - it is economic. Reads are bandwidth-intensive, latency-sensitive, and difficult to meter in a decentralized system. Without a viable economic model for rewarding high-performance reads, most systems avoid the problem entirely, focusing instead on infrequent writes or cold data storage. As a result, truly decentralized applications that depend on fast reads continue to rely on centralized infrastructure.

To ground this challenge in a real-world context, consider video streaming - a canonical example of a read-heavy workload where latency and throughput are paramount. Supporting high-quality 4K video playback requires sustained throughput of at least 40 Mbps, typically delivered in 10~MiB chunks with minimal startup delay. No existing decentralized storage network satisfies this read-throughput requirement (see \cref{tab:read-performance} for a partial list).

Beyond media, modern AI workflows such as retrieval-augmented generation (RAG) systems depend on low-latency access to diverse, high-quality datasets. These models are rapidly moving from research to real-world use: clinicians training diagnostic tools on medical images, governments building domain-specific assistants, and robotics teams teaching home agents new tasks. The expected utility of these systems is large, but realizing it requires two conditions: data must be easy to monetize with fair compensation, and access to that data must be fast and efficient. A data explosion is underway, and while Web3 is uniquely positioned to preserve user control, current infrastructure cannot meet the performance and economic demands of this shift.


\name{} is a decentralized storage network purpose-built for high-performance reads. It combines cryptoeconomic incentives with engineering principles drawn from high-performance Web2 systems to deliver fast, reliable access to data without sacrificing decentralization. \name{} combines a novel auditing protocol, efficient erasure coding, and a dedicated fiber network to support Web2-grade performance at Web3 trust levels.
Its architecture establishes a strong game-theoretic equilibrium in which all participants - storage providers, RPC nodes, and clients - are incentivized to behave honestly and serve data quickly. At its core, the system is built on six tightly integrated pillars, each contributing to its decentralized performance and economic sustainability.

\paragraph{Paid reads - aligned performance incentives.} When users pay for a service, they expect it to be high quality. By requiring payment for reads, \name{} aligns the incentives for users and service providers. Everyone in the ecosystem wants more data to flow: users want a great experience streaming high-quality video or accessing large files and service providers want to serve as much data as they can.

\paragraph{Dedicated fiber network - high-performance network infrastructure.} Unlike the public internet - where congestion, variable latency, and unpredictable routing can degrade performance - a dedicated backbone provides \name{} with consistent bandwidth and low-latency guarantees critical for real-time data access. By deploying storage and RPC nodes directly onto this network, \name{} offers users sub-second access latency and reduces bandwidth costs dramatically - two requirements essential for hot storage and applications like streaming and real-time collaboration. Such a network is critical to reach performance comparable to Web2 cloud providers.

\paragraph{Battle-tested software stack - proven infrastructure for low-latency storage systems.}
\name{}’s software stack is built on engineering principles honed through years of experience developing storage and compute systems for Jump Trading Group's high-performance quantitative research infrastructure and Meta’s global-scale platforms. This includes high-performance I/O pipelines, efficient concurrency primitives, and low level code optimizations. The result is a system architecture that can handle bursty workloads, maintain sub-second level responsiveness, and scale horizontally with minimal overhead. 

\paragraph{Efficient Coding Scheme - high durability with low storage overhead.}
A core challenge in decentralized storage is achieving cost efficiency without sacrificing durability. That is, reducing replication overhead while being able to tolerate and efficiently repair significant amounts of data corruption.
To make this viable, \name{} uses Clay codes~\cite{claycode}, an advanced erasure coding scheme that achieves theoretically optimal repair bandwidth (i.e., Minimum Storage Regenerating) without compromising reconstruction capability (i.e., Maximum Distance Separable). This balance is what enables \name{} to operate with a much lower replication factor than typical Web3 storage networks. While many decentralized systems rely on full replication or less efficient erasure codes - often resulting in 5x to 8x overhead - \name{} achieves durability with a replication factor under 2x. This represents a major cost advantage and brings \name{} close to the 1.2x-1.4x replication range common in high-performance Web2 storage systems.

\paragraph{Incentive-Compatible auditing - economic security without performance trade-offs}
\name{}’s auditing protocol is a hybrid design that combines high-frequency internal audits with low-frequency, cryptographically enforced on-chain verification. Storage providers are continuously challenged to prove possession of random data samples, while their audit behaviors are themselves audited on-chain. This “audit-the-auditor” mechanism allows \name{} to significantly reduce the auditing cost by performing most audits internally, while ensuring that no participant can game the system through collusion or apathy. Importantly, this structure is fully incentive-compatible: honest participation maximizes rewards, while misbehavior leads to slashing and revenue loss. This ensures a strong equilibrium where nodes store data reliably, respond to reads quickly, and verify one another in a decentralized yet economically rational way.

\paragraph{Aptos Blockchain - fast and reliable coordination layer.}
\name{} uses the Aptos blockchain as its coordination and settlement layer. Aptos offers high transaction throughput, low finality times, and a resource-efficient execution model, making it an ideal substrate for managing \name{}’s economic logic. All critical state - including storage commitments, audit outcomes, micropayment channel metadata, and system participation - is recorded and enforced via the \name{} smart contract on Aptos. This allows for decentralized governance, verifiable incentives, and strong fault tolerance without compromising scalability.\\

\begin{table}[t!]
\centering
\begin{tabular}{l|cc|cc|cc}
\toprule
\multirow{2}{*}{\textbf{System}} & \multicolumn{2}{c|}{\textbf{Perf. architecture}} & \multicolumn{2}{c|}{\textbf{User experience}} & \multicolumn{2}{c}{\textbf{Access control}} \\
& \begin{tabular}{@{}c@{}}Ded. fiber\\ network \end{tabular} & \begin{tabular}{@{}c@{}}Replication \\ overhead \end{tabular} & \begin{tabular}{@{}c@{}}4K streaming \\ throughput \end{tabular}  & \begin{tabular}{@{}c@{}}Web2 cost \\ competitiveness \end{tabular} & 
\begin{tabular}{@{}c@{}}Incentivized \\ reads \end{tabular} & Decentralized \\
\midrule
AWS S3 & \cellcolor{green!20}\checkmark & \cellcolor{green!20}est. $1.4x$ & \cellcolor{green!20}\checkmark &\cellcolor{green!20}\checkmark & \cellcolor{green!20}\checkmark & \cellcolor{red!20}\xmark \\
GCS & \cellcolor{green!20}\checkmark & \cellcolor{green!20}est. $1.4x$ & \cellcolor{green!20}\checkmark &\cellcolor{green!20}\checkmark & \cellcolor{green!20}\checkmark & \cellcolor{red!20}\xmark \\
Filecoin & \cellcolor{red!20}\xmark & \cellcolor{red!20}$3{-}6x$ & \cellcolor{red!20}\xmark &\cellcolor{green!20}\checkmark & \cellcolor{red!20}\xmark &\cellcolor{green!20}\checkmark \\
Greenfield & \cellcolor{red!20}\xmark & \cellcolor{yellow!20}$2.5x$ & \cellcolor{red!20}\xmark &\cellcolor{red!20}\xmark & \cellcolor{green!20}\checkmark & \cellcolor{green!20}\checkmark \\
Celestia & \cellcolor{red!20}\xmark & \cellcolor{red!20}$4x$ \cite{celestiareplication}& \cellcolor{red!20}\xmark & \cellcolor{red!20}\xmark & \cellcolor{red!20}\xmark & \cellcolor{green!20}\checkmark \\
Walrus & \cellcolor{red!20}\xmark & \cellcolor{red!20}$4.5x$ \cite{walruswhitepaper} & \cellcolor{red!20}\xmark & \cellcolor{red!20}\xmark & \cellcolor{red!20}\xmark & \cellcolor{green!20}\checkmark \\
Arweave & \cellcolor{red!20}\xmark & \cellcolor{red!20}$15x$ \cite{arweavereplication} & \cellcolor{red!20}\xmark & \cellcolor{red!20}\xmark & \cellcolor{red!20}\xmark & \cellcolor{green!20}\checkmark \\
\textbf{\name{}} & \cellcolor{green!20}\checkmark & \cellcolor{green!20}$<2x$ & \cellcolor{green!20}\checkmark &\cellcolor{green!20}\checkmark & \cellcolor{green!20}\checkmark & \cellcolor{green!20}\checkmark \\
\bottomrule
\end{tabular}
\caption{Qualitative comparison of storage systems across performance, user experience, and decentralization properties. \name{} is the only fully decentralized protocol with the performance architecture needed to deliver Web2-grade read experiences at competitive cost.}
\label{tab:read-performance}
\end{table}

\noindent
Building on the earlier video streaming example, we now demonstrate how \name{}’s design enables real-world performance.
A creator uploads a high-resolution video using the \name{} client SDK, which encodes the file with Clay erasure codes and submits a cryptographic commitment to the Aptos blockchain. The client works with an RPC node to distribute encoded chunks across globally dispersed storage providers connected via a dedicated network backbone.
When a viewer presses “play,” the RPC node fetches the necessary chunks, reconstructs the video on the fly, and streams it with low latency. Micropayments flow in real time to compensate both storage and RPC providers. In the background, the auditing protocol ensures that data remains intact, available, and verifiably stored.

Unlike Web2 walled gardens, where creators relinquish control over content and monetization, \name{}-backed streaming dApps enable creator-owned business models. With fast reads and native micropayments, creators can embed custom interstitials, run their own license servers, or enforce playback restrictions via DRM. This unlocks new monetization paths, from direct sponsorships to platform-independent distribution, while infrastructure providers are rewarded proportionally to their contributions.

In short, \name{} brings Web2-grade performance to Web3 infrastructure - not by compromise, but through principled engineering, incentive-aligned economics, and a deep understanding of what it takes to build a cost-efficient, highly scalable system.

\vspace{0.5em}Section~\ref{sec:SystemArchitecture} presents a high-level picture of \name{}'s architecture,
while Section~\ref{sec:systemprimitives} goes into more detail on a selected
number of important primitives, such as micropayment channels, erasure coding, and data preparation. Sections~\ref{sec:AuditAndIncentiveCompatibility} and~\ref{sec:Economics} discuss \name{}'s hybrid audit scheme and
economic incentives, respectively. Finally, Section~\ref{sec:ComputeUseCases}
presents some future use-cases for \name{}.

\section{System Architecture}
\label{sec:SystemArchitecture}

\name{} is comprised of four service layers: the client SDK, the RPC node layer, the storage provider (SP) layer, and the coordination layer. Figure~\ref{fig:systemdiagram} shows an example system diagram.

Clients store their data using the SDK. It prepares data for storage and transfer it to an RPC node, which manages the write process. Clients can also set up micropayment channels for reading data from \name{} using the SDK.

Data transfers and end-user payments all occur at the RPC node layer. RPC nodes encode and disperse
user data to the storage node layer during writes and gather and decode
data during reads. 

The SP layer is responsible for storing data on behalf of users. SPs also conduct peer-to-peer audits of data. Storage
providers are built with very large amounts of local storage. \name{} provides application-level
data integrity, so SP operators do not need to provision RAID arrays
or other data integrity solutions.

Finally, \name{} coordinates internal state via the Aptos blockchain. The \name{} smart contract
manages several functions on behalf of the system: SP participation, data
placement and lifecycle, cryptoeconomic security rewards and punishments, as well as payment flows.

\subsection{User Data: \userfile{}s, \sectionuserfile{}s, \spunitofstorage{}s, and \proofforaudit{}s}

\begin{itemize}
    \item \textbf{\userfile{}s} are arbitrary-sized Binary Large Objects such as images or videos.
    \item \textbf{\sectionuserfile{}s} are fixed-size portions of a \userfile{}, roughly \sectionuserfilesize{}.\footnote{Files that are smaller than roughly \sectionuserfilesize{} are
zero-padded. As a result, very small files incur overheads, and we expect
clients will store them together to lower those overheads.}
    \item \textbf{\spunitofstorage{}s} are fixed-size (roughly \spunitofstoragesize{}) portions of a \sectionuserfile{}, formed by erasure-coding.
    \item A \textbf{\proofforaudit} is a small (around \proofforauditsize{}), fixed-size portion of a \spunitofstorage{} for use in auditing.
\end{itemize}

The Client SDK reads and writes \userfile{}s, but also allows byte range reads for efficiency.
Internally, these \userfile{}s are
partitioned and specially prepared for storage on the set of storage providers. Users pay to store data
for a certain duration and \userfile{}s stored in \name{} are immutable.

\subsection{Client SDK}

Clients use the SDK to perform functions such as establishing and managing payment channels, preparing data for upload, and writing and reading data from \name{}. The Client
SDK contains a CLI and libraries able to integrate into user-facing dApps.

\paragraph{Writing data.} When writing data, the \name{} Client SDK erasure codes the \userfile{} and computes cryptographic commitments on the encoded data. The SDK then
stores \userfile{} metadata on the coordination layer by calling smart contract functions. The user pays to store the \userfile{} for a specific duration. Once the coordination layer has stored the metadata and transferred payment, the SDK 
contacts an RPC node to transfer data. Once the RPC node finishes the write procedure, it changes the \userfile{} metadata on the coordination layer, marking it stored
in \name{} and ready for reading.

\paragraph{Reading data.} The first step for reading from \name{} is to use the Client SDK to establish a micropayment channel with the RPC node. Thereafter, clients sign micropayment transactions in the channel mixed with data reads from the RPC node.

\subsection{RPC Nodes}

RPC nodes are the gateway to the \name{} storage system. They offer an API which, in concert with the Client SDK, enables clients to write and read data. Furthermore, the RPC nodes call functions on the \name{} smart contract on the coordination layer to read and write \userfile{} metadata. Finally, RPC nodes call the \name{} SP API to write and read \spunitofstorage{}s of data from the storage providers.

\paragraph{Writing data.} During the \userfile{} write procedure, the RPC node checks cryptographic commitments in the \name{} metadata and sends the encoded \spunitofstorage{}s to the SP layer. The specific SP nodes are
assigned by the smart contract. After the SPs have acknowledged the data is received, the RPC node uses the coordination layer to mark the \userfile{} ready for reading.

\paragraph{Reading data.} RPC nodes pay for reading data from SPs. When joining the network, RPC nodes establish payment channels
to the SP layer. RPC nodes read \spunitofstorage{}s from SPs and
reconstruct the requested range of \userfile{} data. Since \userfile{}s are cryptographically committed, attempts to alter data will be detected.

\begin{figure}
  \includegraphics[width=\linewidth]{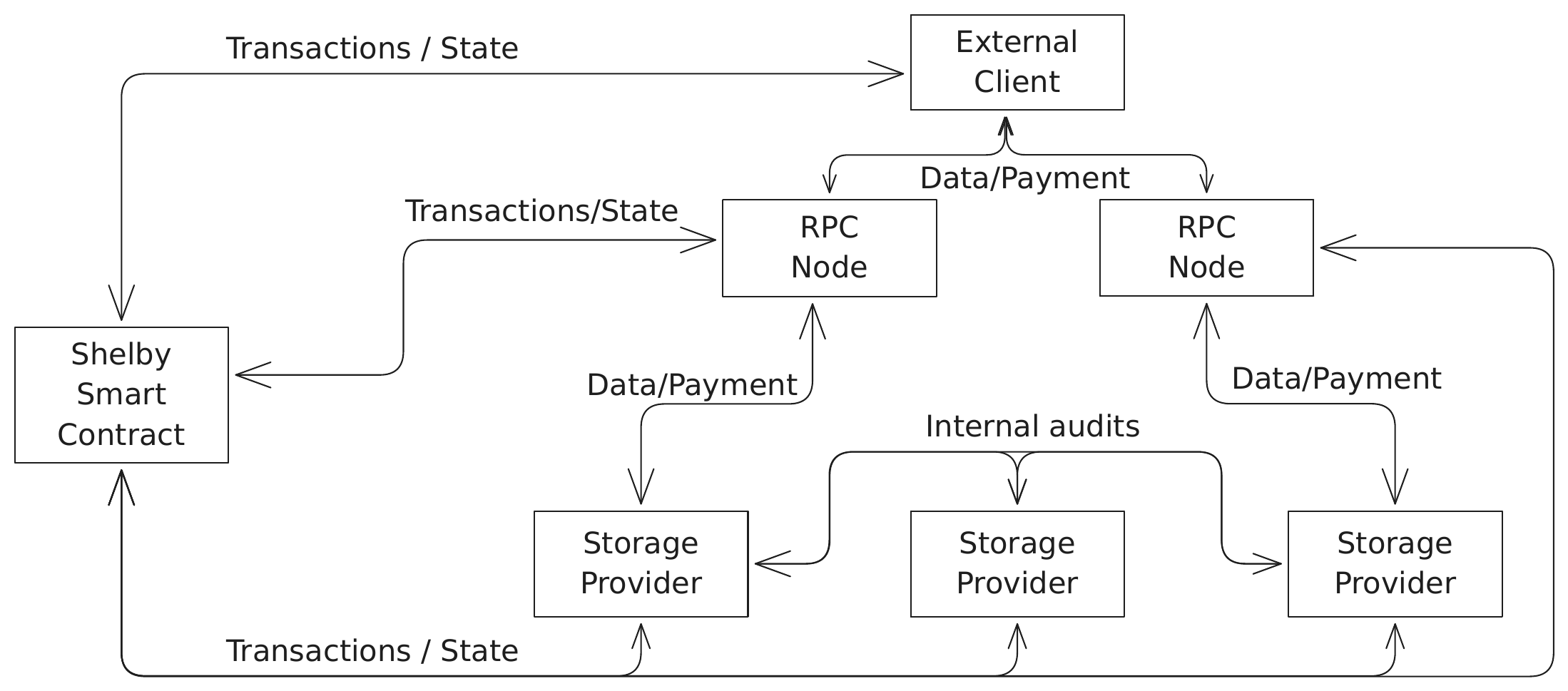}
  \caption{System Diagram}
  \label{fig:systemdiagram}
\end{figure}

\subsection{Storage Providers}

Storage providers are the core of \name{}. An SP stores data on behalf of users and audits other storage providers to ensure data durability and availability. Similar to RPC nodes, SPs use the coordination layer to follow \name{} metadata and call smart contraction functions.

\paragraph{Writing data.} During data writes, SPs track the state of the \name{} smart contract to ensure they know what data \spunitofstorage{}s are assigned to them.

\paragraph{Reading data.} During data reads, SPs respond to paid requests for data by returning the requested \spunitofstorage{} and cryptographic commitment.

\paragraph{Auditing data.} Periodically, SPs are required to conduct a
series of audits. As auditor, an SP requests a \proofforaudit{} and
cryptographic commitment from an auditee storage provider. The auditee responds with the requested \proofforaudit{} and the proof of data. Auditors store proofs locally in case their audit records are selected for verification by the \name{} smart contract.
Audit records that are selected for verification by the \name{} smart contract are submitted directly to the \name{} smart contract using the coordination layer. More details are given in Section~\ref{sec:AuditAndIncentiveCompatibility}.

\subsection{\name{} Smart Contract}

Deployed on the Aptos blockchain, the \name{} smart contract assigns \spunitofstorage{}s to SPs during the \userfile{} write procedure. It manages \userfile{} metadata including lifecycle and accepts payments from users. It sets audit schedules for data and SPs. And ultimately, the smart contract contains the logic to punish bad actors in the system when they have provably acted in a byzantine fashion (either by crashing, censoring, or lying about data).

\paragraph{Writing data.} During writes, \name{} a client submits payment
and a succinct cryptographic commitment of their \userfile{} to the smart contract. The smart contract randomly assigns
\spunitofstorage{}s to SPs, and this mapping is stored in the \userfile{} metadata. Once the \spunitofstorage{}s are accepted by the SPs, the \userfile{} is marked ready by the RPC node. After this point in time, the data is durably stored and performantly retrievable.

\paragraph{Reading data.} \name{}'s high throughput data reads do not require any fine-grained coordination from the blockchain layer. Rather, micropayment channels allow for scalable system performance with loose global coordination, a key technique of high performance computing. Clients can also use micropayment channels with the RPC node layer, although RPC nodes may choose to support many payment mechanisms to encourage \name{} usage.

\section{System Details}
\label{sec:systemprimitives}

This section presents some of the low-level system details used to
construct \name{}. While not required for the presentation of the high-level
system architecture in Section~\ref{sec:SystemArchitecture}, these details
enable \name{} to hit the desired system characteristics of low cost 
and high performance, durability, and availability.

\subsection{Dedicated Network Layer}
\label{sec:dznetwork}

Web2 competitive price-performance requires Web2 physical infrastructure. Web2 providers use private, dedicated high-performance networks for efficient data transfer at reasonable cost. Similarly, \name{} uses a dedicated high-performance network layer to achieve price-performance not possible for Web3 operating over the public internet. 
\name{} uses this layer
to achieve a Byzantine fault tolerant system without requiring application-layer reliable 
broadcast mechanisms over the public Internet. DoubleZero~\cite{doublezero} recently emerged as a decentralized option for a suitable network layer.

\subsection{Micropayment Channels}
\label{sec:paymentchannels}

\name{} is designed for high-throughput, low-latency data serving. As such,
RPC nodes and storage providers (SPs) cannot mix on-chain payments while
serving data. There are many off-chain payment options for end users
(credit cards, monthly subscriptions, etc.), and RPC nodes can offer the best
different payment mechanism for their user community.

As a default, \name{} offers micropayment channels~\cite{wikicontracts} 
as a payment mechanism between clients and RPC nodes
and RPC nodes and SPs. Micropayment channels enable
off-chain transfers to occur in an optimistic manner allowing high
scalability. Since on-chain transactions are required only for creation and settlement, micropayment channels also allow for small payments with low overhead cost.

While the Aptos blockchain offers industry-leading low-latency consensus, micropayment channels provide an additional performance and cost advantage for high-frequency, peer-to-peer payments. Because they bypass consensus entirely, micropayment channels support payments to be freely mixed with data reads, which can be an essential feature when small reads must be extremely cheap to support responsive, usage-based incentives.

Very briefly, a basic unidirectional micropayment channel between a client and a server
requires funds to be transferred from the
client to a multi-signature account signed by both parties. The server also sends the client a refund transaction which the client can use to regain its funds after
a certain amount of time. The client makes payments by sending new refund transactions
with a slightly smaller refund to the client and a slightly earlier allowed settlement time. 
If either party stops cooperating, the most recently
signed refund transaction can be used to settle the channel before the other refund transactions become valid.

Micropayment channels are not perfectly trustless, but uncooperative parties
will quickly stop being paid for invalid service. Hence, the expected value at risk is
small.

\subsection{Erasure Coding}
\label{sec:erasurecoding}
\label{sec:SLAs}

\name{} uses Clay codes~\cite{claycode} (Coupled-Layer codes) to prepare
user data for storage on SPs. The choice
of an erasure coding scheme impacts system performance characteristics
such as write and read speed, data storage durability and availability.
Web2 storage systems are attractive to users because they offer high durability (data has an extremely low likelihood of data loss)
and availability (data can be read \textit{right now}) at reasonable costs and performance levels.
For \name{} to achieve the same levels of data storage durability and availability
requires an erasure coding scheme fit for the task. (\cite{codesurvey} surveys modern erasure codes.)


Erasure-coding schemes divide data into $k$ blocks and encode it into $n$ blocks (where $n>k$) to add redundancy. Storage systems then store the $n$ erasure-coded
blocks onto different nodes.
The ideal erasure coding scheme for a storage system minimizes the bandwidth and the storage for a given durability. That is, the ideal coding scheme is Minimum Storage Regenerating~\cite{regencode}\cite{claycode}\cite{Rashmi2015Cake} (MSR) and Maximum Distance Separable~\cite{jin_new_2024} (MDS). Clay codes are a practical family of codes with both properties. During data repair, they exhibit 60\% less bandwidth usage compared to Reed-Solomon~\cite{reedsol} codes. Clay codes also maintain the MDS property where any $k$ out of $n$ total erasure-coded portionare sufficient
to fully decode the data. Real-world storage systems such as Ceph~\cite{weil2006ceph} have performance benchmarks that demonstrate the 
advantages of Clay codes.

Clay codes are ideal for \name{} because they balance performance, durability,
and availability characteristics:
\begin{itemize}
  \item \textbf{Computationally efficient}: Clay codes have efficient encoding and decoding, avoiding excessive CPU time as a performance degrading bottleneck.
  \item \textbf{High performance for target data size}: \name{} is designed to work well for \userfile{}s of at least \sectionuserfilesize{}. Clay codes work best when the data being coded is sufficiently large to be divided into the required number of sub-units without excessive overhead. As such, \name{} has good
  storage efficiency for the targeted data size.
  \item \textbf{Flexible Repair Structure}: Clay codes can achieve optimal repair band provide optimal repair efficiency when using a specific ordering of repair actions. When the optimal repair pattern cannot be followed, \name{} can fall back to the MDS property (where any $k$ chunks can recover data) even if it must temporarily sacrifice repair bandwidth efficiency.
  \item \textbf{Repair Coordination}: \name{}'s coordination layer allows planning for bandwidth-optimal recoveries, a key feature of Clay codes.
\end{itemize}

By using Clay codes to store data, \name{} offers 99.999999999\% (11 nines) data durability, and 99.9\% availability. Derivations of the
durability and availability of \name{} are given in Appendix~\ref{sec:erasurecodingappendix}.

\subsection{Cryptographic Commitments}
\label{sec:datacommitments}

\name{} uses a Vector Commitment (VC) scheme to enable data verification. In~\cite{catalano13}, the authors state
``VCs allow to commit to an ordered sequence of $q$ values $(m_1, \ldots, m_q)$ in such a way
that one can later open the commitment at specific positions (e.g., prove that $m_i$ is the 
$i$-th committed message)."
In \name{}, we use Merkle trees~\cite{merkle1987digital} as the VC scheme to bind the encoded \spunitofstorage{}s
to a root hash. Thereafter, responses to requests for data include the data validity proof.

Data integrity is a critically important feature for a storage system. Subject to standard cryptographic caveats (i.e., that no known computing method
can efficiently find hash collisions), \textbf{it is impossible for components of \name{} to alter user data without being detected.}

The choice of VC also allows for efficient auditing of SPs.
The
root hash of the Merkle tree is stored in the \userfile{} metadata using the \name{} smart contract.
Thereafter, \proofforaudit{}s and inclusion proofs are used to verify the data by auditors.
The computational efficiency of a Merkle tree inclusion proof allows
for audits to be computed on-chain. Section~\ref{sec:AuditAndIncentiveCompatibility} has more details on the hybrid auditing scheme.


\subsection{Engineering for High Performance}

\name{} is built using the experience and engineering techniques that underpin
the High Performance Compute (HPC) storage system used by Jump Trading Group. HPC engineering
ensures that the system components are optimized to an equally
high degree. Amdahl's law requires that for an
entire computation to be optimal, all parts of it must
be optimal. For distributed systems, this optimization
process can leverage a multitude of technologies at
each layer of the software application stack:

\begin{itemize}
    \item \textbf{Network:} The use of multicast
allows for efficient broadcast of data among the system
nodes. The application layer need not concern itself
with duplicating data, as the network switch hardware
is responsible for ensuring the data is delivered correctly
to all members of the multicast group.

\item \textbf{Packet processing:} Kernel-bypass networking (e.g. DPDK~\cite{dpdk}) allows for Linux systems
to process packets at much higher rates than through
traditional kernel-based techniques. Meanwhile,
the eXpress Data Path~\cite{xdp} (XDP) in the Linux kernel
offers high performance packet processing without
needing to fully bypass the kernel. These techniques
can be used to break through bottlenecks caused by
packet processing overheads.

\item \textbf{Storage:} \name{} SPs will seek to maximize their profit.
As a result, many or most SPs will use
spinning hard drives. While hard drives
are not known for extremely high performance,
careful planning of data placement and ordering
of hard drive seeks to reduce seek time help
maximize performance.

\item \textbf{CPU:} Modern CPUs have high core counts and, as such, act
as tiny distributed systems with internal networks~\cite{kumar2002network}. As such, these
networks can become congested if dataflow is not
carefully designed. Further, the physical design of high core count CPUs requires
tradeoffs to be made around memory latency. On-chip cache and main system memory exhibit different latency characteristics due to Non-Uniform Memory Access (NUMA). 
In particular, multi-socket motherboards are incapable of
offering the same latency to main memory from all sockets,
due to the speed of light in Printed Circuit Board (PCB)
wire traces.
For all these reasons, topology-aware
application design is thus an important concept
to consider for HPC engineering.

\item \textbf{Memory efficiency:} Avoiding unnecessary overheads for data sharing
is important for achieving full CPU utilization. The
use of lock-free data structures that avoid false sharing
of CPU cache enables task parallel processing to
make optimal use of on-chip resources.

\item \textbf{Data parallelism:} Modern CPUs have significant on-chip
resources devoted to data parallel computation. By designing applications
for data parallelism from the beginning, the full usage of vectorized
instructions unlocks the power of modern CPUs.

\item \textbf{Erasure coding acceleration:}
Modern CPUs have accelerated instructions for the Galois field modular arithmetic required to perform erasure coding, and there are advanced algorithms~\cite{taffet23} which can attain encode/decode rates that are in excess of the bandwidths available on network cards.

\item \textbf{\name{}-internal communication:} Distributed systems comprised of individually optimized components
need optimal inter-node communication patterns as well.
HPC techniques ensure that dataflows avoid local bottlenecks
which compromise global system performance. High performance systems use ``request hedging," in which requests are sent to all servers that can satisfy the request.
The first response is used, and the remaining requests are canceled and in-flight responses are ignored.
This approach wastes some amount of resources, but can prevent temporary hotspots from dramatically affecting the tail latency of a system.

\end{itemize}

The use of HPC principles in \name{} 
enables read performance to scale with the available hardware. The throughput of the system
is dictated by the aggregate bandwidth of the RPC node layer.

\subsection{Data Preparation Example}
\label{sec:datapreparation}

Figure~\ref{fig:dataprep} shows an example data preparation of
a user file. The Client SDK partitions the \userfile{} into
\sectionuserfile{}s (if the \userfile{} is larger than a single
\sectionuserfile{}). Since it is unlikely the \userfile{} will
be partitioned perfectly into the
\sectionuserfilesize{}s, the final \sectionuserfile{}
is zero-padded up to the correct length. Then, each
\sectionuserfile{} is erasure coded with Clay codes, which forms some number of
\spunitofstorage{}s. Each of those \spunitofstorage{}s
is then cryptographically committed, as is 
the full \userfile{}.
All of these cryptographic commitments are included
in the transaction the Client SDK commits to the Aptos blockchain
during the write procedure.

\begin{figure}
  \includegraphics[width=\linewidth]{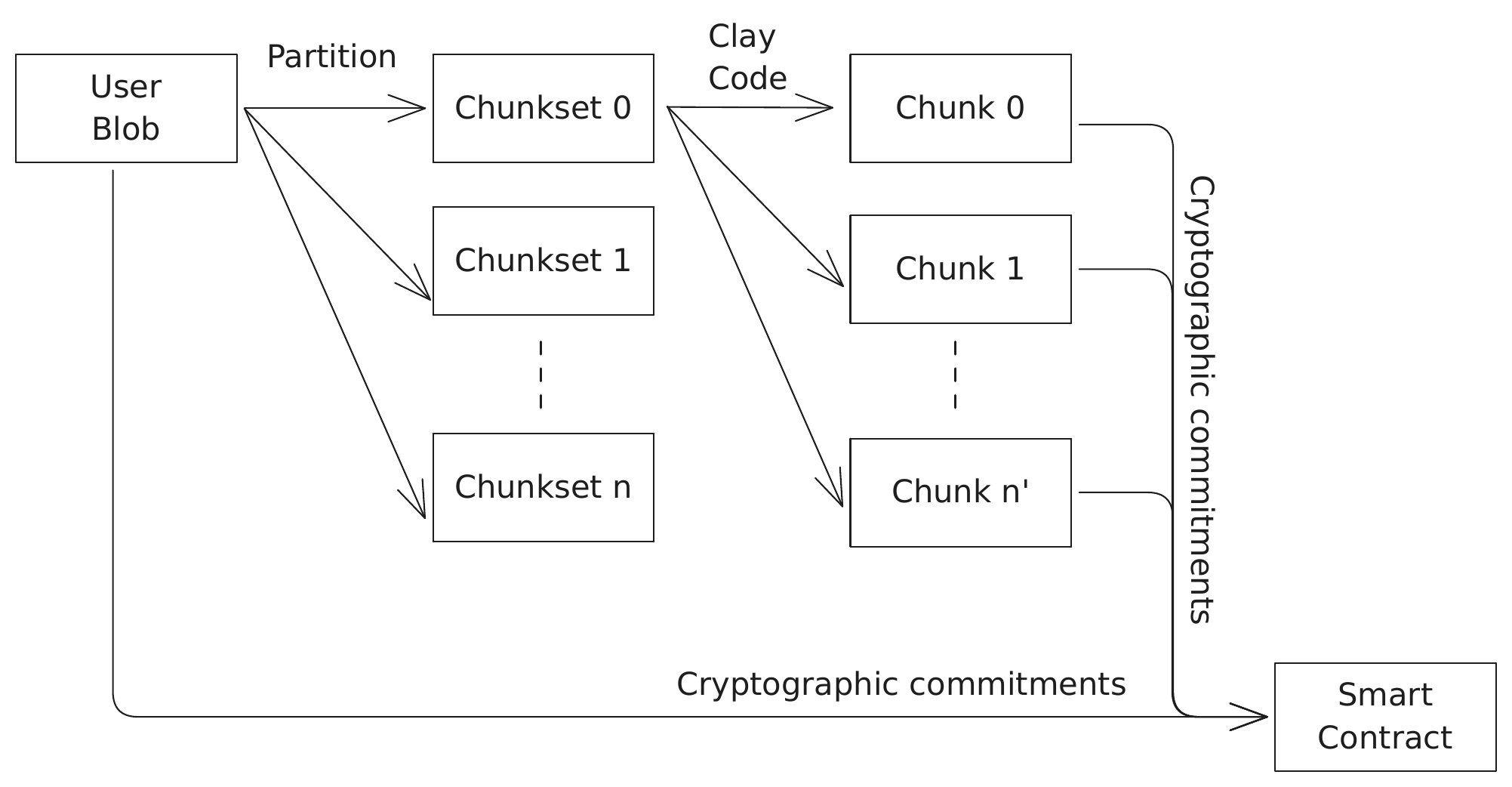}
  \caption{Data preparation by the Client SDK}
  \label{fig:dataprep}
\end{figure}

\section{Audit and Incentive Compatibility}
\label{sec:AuditAndIncentiveCompatibility}

\paragraph{Overview.}
The audit subsystem in \name{} is designed to ensure that storage providers (SPs) reliably store the data they have committed to, and to verify that peer auditing activity (i.e., SPs auditing each other) is conducted honestly. The protocol follows a hybrid design that combines high-volume, low-cost \textit{internal audits} with low-frequency \textit{on-chain audits} that also serve to audit the auditors. This architecture enables strong guarantees for both Byzantine Fault Tolerance (BFT) and Incentive Compatibility (IC), while keeping cost and performance overheads tractable at scale.

In \name{}, \textit{internal audits} are lightweight, peer-to-peer checks performed off-chain, while \textit{on-chain audits} are verifiable challenges enforced by the blockchain smart contract.
Ths hybrid audit design is motivated by the complementary limitations of purely on-chain and purely internal auditing. On the one hand, on-chain audits provide strong cryptoeconomic guarantees, but they are inherently constrained by gas costs, bandwidth requirements, and the system-wide consensus overhead of executing audits through the Aptos validator set. These constraints make it impractical to run on-chain audits at the volume needed to enforce good behavior at scale. On the other hand, internal auditing mechanisms, while cheap and performant, are vulnerable to rational deviations. Without an external enforcement layer, SPs can adopt collusive strategies - such as universally reporting success or ignoring audits entirely - that maximize local utility while undermining system integrity. \name{}'s hybrid design seeks to reconcile these trade-offs by using internal audits to detect most misbehavior and on-chain audits to impose accountability where it matters.

All actively punitive actions in the system - such as slashing misbehaving SPs or penalizing false audit reports - are carried out on-chain, which protects SP minorities from being slashed by a majority. The audits procedure is divided into predetermined audit epochs. At the end of each audit epoch, the blockchain performs two key roles:
\begin{enumerate}
  \item It audits SPs whose behavior appears suspicious based on internal audit outcomes, issuing randomized challenges to verify storage correctness.
  \item It audits the auditors by requiring them to reproduce a random subset of the audit proofs they claimed to have verified, thereby validating the integrity of internal audit attestations.
\end{enumerate}

These on-chain audits are designed to be sparse but unpredictable, serving as a forensic tool rather than a primary enforcement mechanism. By applying external scrutiny only to a strategically chosen subset of actors and actions, the system achieves robustness under standard Byzantine Fault Tolerance and Game Theory assumptions, while preserving scalability and economic efficiency.

\subsection{Internal Auditing Protocol}

Internal audits in \name{} are designed to operate at high frequency and low cost, enabling continuous verification of SPs behavior without incurring the overhead of consensus-layer operations. The mechanism relies on publicly verifiable randomness from the Aptos blockchain to assign storage challenges to SPs throughout each audit epoch. Each challenge corresponds to a small, randomly selected portion of a stored \spunitofstorage{} - typically a \proofforauditsize{} \proofforaudit{} -- whose possession must be demonstrated by the responsible SP.

Upon receiving a challenge, the designated SP generates a succinct cryptographic proof of possession (e.g., a Merkle proof) and broadcasts this proof to the rest of the network%
\footnote{In practice, broadcast is implemented using scalable multicast protocols (e.g., to subcommittees), but for simplicity we describe the mechanism in terms of full broadcast.}.
All SPs maintain a local \textit{scoreboard}, which records whether each peer has responded correctly to its assigned challenges. These scoreboards serve as the basis for reward distribution and act as the primary input for determining which SPs are subject to further on-chain audits.

To ensure accountability and auditability, each SP retains all audit proofs it receives for the duration of two audit epochs. At the end of each epoch, SPs publish a compressed version of their scoreboard to the blockchain. Each scoreboard is comprised of $(n-1)$ bit vectors, where $n$ is the number of SPs and the length of the $i$'s vector is the number of audits that SP $i$ was challenged with during the epoch. In typical operation, these bit vectors are dominated by successful audits and thus exhibit high regularity - either sparsity or uniformity - making them amenable to efficient compression. This allows the full scoreboard to be submitted on-chain with modest bandwidth and gas costs.

Incentives for participating in internal audits are structured as follows. Rewards for the audit epoch are divided into two pools:
\begin{itemize}
  \item \textbf{Storage rewards}, which are proportional to the volume of data stored and scaled by the provider’s audit score. The audit score reflects how consistently the SP has responded correctly to internal audit challenges during the epoch.
  \item \textbf{Auditor rewards}, which provide a small, fixed payment for each successful audit reported by the SP. These rewards are independent of how much data the SP stores. 
\end{itemize}

Let $score_i\in [0,1]$ denote the audit score for SP~$i$, derived from the set of bit vectors submitted by its peers. To mitigate the influence of Byzantine or otherwise erratic behavior, the highest and lowest third of evaluations are discarded prior to aggregation - a standard resilience technique under the assumption that at most one-third of nodes may be faulty.
The remaining evaluations are aggregated (e.g., via majority vote or weighted average) to compute $score_i$. The allocated storage rewards are then scaled by $score_i$. Additionally, if SP~$i$ was assigned $a_i$ audits as an auditor and successfully verified $s_i$ of them, it receives $s_i \cdot \textit{rwd}_{\textit{au}}$ in audit rewards,  where $\textit{rwd}_{\textit{au}}$ denotes the fixed reward for each successfully reported audit (i.e., a `1' entry in the scoreboard).

This structure positively incentivizes SPs to (a) store their assigned data and (b) perform audits faithfully and in a timely fashion. Importantly, it does not create a zero-sum game among SPs: one provider’s underperformance does not automatically increase rewards for others. The system rewards independent correctness rather than relative rank.

\subsection{On-Chain Auditing of the Auditors}

The on-chain auditing mechanism in \name{} serves as a low-frequency, high-integrity enforcement layer to ensure honest participation in internal audits. It targets both storage providers (SPs) suspected of underperformance as well as verifying auditors correct reporting. While the majority of auditing activity and scoring is conducted off-chain, \emph{actively punitive measures} - such as slashing - are triggered on-chain and rely on cryptographic evidence. These slashing events may originate either from on-chain audit challenges or from SPs submitting verifiable evidence of protocol violations encountered during internal audits.

At the end of each audit epoch, the Aptos blockchain uses its native randomness to assign two types of on-chain audit challenges:
\begin{enumerate}
    \item \textbf{Auditee audits}: SPs with low audit scores are subject to direct on-chain challenges testing their storage integrity. For an SP with score $\textit{score}_i \in [0,1]$, the number of challenges assigned is $(1 - \textit{score}_i^2) \cdot C$, where $C$ is a configurable system parameter. This quadratic penalty scheme ensures that well-performing providers are audited infrequently, while those with poor records face increasing scrutiny.
    
    \item \textbf{Auditor audits}: SPs are challenged to reproduce audit proofs corresponding to a randomly selected subset of the $1$ entries in their published scoreboards. Each such entry asserts that the SP verified a peer's proof during the internal audit phase. Failure to produce a valid and matching proof results in slashing for false reporting.
\end{enumerate}

Auditee challenges are constructed by sampling random \proofforaudit s from the SP’s declared holdings, and each challenge requires a cryptographic proof of possession - typically a Merkle membership proof. These proofs are submitted on-chain and verified by a smart contract. Auditor audits are symmetric in structure: the challenged auditor must submit the audit proof it previously claimed to have verified and stored. The contract checks consistency against the scoreboard and validates the cryptographic correctness of the proof.

In addition to these scheduled on-chain audits, SPs are also permitted to submit evidence of invalid audit responses observed during internal operation. If an SP detects a malformed or provably incorrect audit proof from a peer, it may post this evidence on-chain to initiate slashing. In such cases, the reporting SP receives a portion of the slashed funds as a reward, aligning incentives for active monitoring and detection. Moreover, since honest participation requires only local computation and storage, the cost of compliance remains low relative to the risk and penalty of deviation.

All on-chain audits are probabilistic and non-deterministic from the perspective of SPs. This unpredictability, coupled with the possibility of peer-initiated slashing, ensures that rational SPs have strong incentives to store data, retain audit proofs, and participate honestly in internal audits. Because enforcement relies on objectively verifiable data and does not depend on subjective peer consensus, the protocol avoids vulnerabilities associated with collusion or false majorities.

This mechanism disrupts incentive-incompatible equilibria in which SPs collude or neglect their responsibilities. Even if the likelihood of on-chain auditing is low, the cost of failing such an audit - or being reported for provable misbehavior - exceeds the cost of honest participation, making the protocol robust under standard rational and Byzantine assumptions.

\subsection{Byzantine Fault Tolerance of Audit Score Computation}
\label{sec:BFTCorrectness}

The goal of internal auditing in \name{} is to produce an audit score for each SP that accurately reflects its storage behavior, despite the presence of faulty or Byzantine nodes. This subsection outlines the main intuition for BFT correctness. 

A standard BFT model assumes a partially synchronous network and a set of $n$ SPs, of which at most $f < n/3$ may behave arbitrarily (i.e., be Byzantine).  
Correct SPs follow the internal audit protocol as specified.

\paragraph{BFT Intuition.}
Recall that at the end of each audit epoch, the audit score for SP~$j$ is computed from the scores it received from its peers by trimming the highest $f$ and lowest $f$ peer evaluations (inside a peer evaluation a missing `1' equals a `0'). The remaining scores are then aggregated (e.g., via majority or averaging) to obtain a normalized score $\textit{score}_j \in [0,1]$.
Since at most $f < n/3$ SPs are faulty, trimming the top and bottom thirds of evaluations guarantees that the aggregated score is within the range of the highest and lowest honest scores given to $j$.
Honest SPs report audit outcomes based on actual verification, and network noise is modeled within the faulty set.  
Thus, for a correct SP~$j$, the computed score $score_j$ will closely approximate its true success rate in responding to audit challenges. \ 
Faulty SPs, by contrast, will tend to have persistently lower scores, either due to audit failures, missing responses, or detection of incorrect proofs.

\paragraph{Summary.}
Under standard BFT assumptions, the internal auditing mechanism ensures that:
\begin{itemize}
    \item Honest SPs receive audit scores close to $1$, reflecting faithful storage behavior.
    \item Faulty SPs receive scores reflecting their true audit failures, and cannot inflate their success rates.
\end{itemize}

\subsection{Game-Theoretic Incentive Compatibility}
\label{sec:GMIC}
Establishing incentive compatibility (IC) in decentralized storage systems is nontrivial under standard rationality assumptions.
Adopting a more cautious approach requires a list of modeling assumptions
and select key justifications for the incentive compatibility arguments. Additional details with proof sketches are in \Cref{app:IC}.

\paragraph{Rationality.}
Storage providers (SPs) are assumed to be expected-utility maximizers. Each SP selects its strategy in order to maximize expected payoff, accounting for all rewards, penalties, and operational costs associated with the protocol.

\paragraph{Audit Verifiability.}
Given an audit proof, its correctness (i.e., whether it corresponds to a valid \proofforaudit{} of the committed \spunitofstorage) can be verified efficiently and unambiguously by auditors and by on-chain smart contracts.
In \name{}, Merkle proofs are used as the underlying verifiable structure.

\paragraph{Audit Proofs Imply Persistent Storage.}\label{ass}
An SP producing a valid audit proof for a challenged \spunitofstorage{} retains the corresponding data locally, except with negligible probability. This assumption is justified by the following lemma.

\begin{restatable}{lemma}{LEMstorageIncentives}[Economic Incentives for Persistent Storage]
\label{lemma:storage-incentives}
Let $c_s$ be the per-epoch cost of storing a chunk, $c_r$ the cost of retrieving it externally and $p_a$ the per-epoch probability of an audit.%
\footnote{Conservatively assume that a misbehaving SP always succeeds in performing a timely external retrieval. This strengthens the result since, in practice, the misbehaving SP is also taking the considerable risk of failing to externally retrieve the data in a timely manner.}
If \
$p_a \cdot c_r \ge c_s$,
then the rational strategy for an SP is to retain its assigned \spunitofstorage{} locally rather than deleting it and attempting on-demand retrieval.
\end{restatable}


Numerical estimates (see \Cref{sec:RealityBasedIncentives}) suggest that with realistic storage and retrieval costs, and under reasonable audit frequencies, the above inequality is satisfied with a substantial margin.

\paragraph{Reward/Penalty Recap.}
For convenience, briefly recall the relevant reward and penalty structures:\\[0.5ex]
    \textit{Storage rewards:} Proportional to stored volume, scaled by the SP’s audit score.\\[0.2ex]
    \textit{Auditor rewards:} An SP earns a reward $\textit{rwd}_{\textit{au}} > 0$ for each successful audit reported.\\[0.2ex]
    \textit{Audit-the-auditor penalties:} A 1-entry in a scoreboard is independently selected for verification w.p.~$p_{\textit{ata}} > 0$. Failure to produce the corresponding audit proof results in a slashing penalty $S_{\textit{ata}}$.\\[0.2ex]
    \textit{Slashing rewards:} SPs who submit verifiable proofs of peer misbehavior (e.g., invalid proofs) are rewarded with $r_{\textit{slash}} > 0$.\\[1ex]
\indent Denote by $\pi$ an audit-challenge proof consisting of (i) the challenged \proofforaudit{} (a randomly selected \proofforauditsize{} \proofforaudit{}) and (ii) a succinct cryptographic proof verifying its correctness. $p_{\textit{inv}}(\pi)$ is defined as the probability that the proof $\pi$ is invalid. To incentivize auditors to ensure that $p_{\textit{inv}}(\pi) < \varepsilon$ for a small threshold $\varepsilon > 0$, the slashing penalty $S_{\textit{ata}}$ should satisfy:
\[
    S_{\textit{ata}} \geq \frac{\textit{rwd}_{\textit{au}}}{p_{\textit{ata}} \cdot \varepsilon}.
\]

\begin{definition}[Honest Strategy]
\label{def:HonestStrategy}
Let $\sigma_i^{\textit{honest}}$ denote the strategy of storage provider $i$ during an audit epoch, consisting of two roles:

\begin{itemize}
    \item \textbf{Auditee behavior:}
    \begin{itemize}
        \item \textbf{Persistent Storage:} Store all assigned \spunitofstorage{}s faithfully throughout the epoch.
        \item \textbf{Audit Proof Generation:} Upon receiving an audit challenge, compute the proof $\pi_i$ and broadcast $\pi_i$ to peers (or the assigned subcommittee).
    \end{itemize}
    
    \item \textbf{Auditor behavior:}
    \begin{itemize}
        \item \textbf{Audit Verification and Reporting:} For each received proof $\pi_j$:
        \begin{itemize}
            \item Record success (`1') in the scoreboard if
            $
            p_{\textit{inv}}(\pi_j) \leq \varepsilon,$ 
            where $\varepsilon > 0$ is the given threshold; otherwise record failure (`0').
            \item Retain $\pi_j$ for audit-the-auditor verification.
        \end{itemize}
        \item \textbf{Slashing Proof Submission:} Submit invalid proofs $\pi_j$ on-chain as slashing evidence.
        \item \textbf{Publishing Scoreboard:} At the end of the epoch, publish the complete peer audit scoreboard.
    \end{itemize}
\end{itemize}
\end{definition}

There are three important IC properties that \name's Auditing scheme satisfies.

\paragraph{(1) Nash Equilibrium of Honest Strategy.} 
\name{} is IC in the most basic sense, as captured by \Cref{thm:honest-nash} below.

\begin{restatable}{theorem}{honestNash}[Incentive Compatibility of the Honest Strategy Profile]
\label{thm:honest-nash}
Let $\vec{\sigma}^{\textit{honest}} = (\sigma_1^{\textit{honest}}, \ldots, \sigma_n^{\textit{honest}})$ denote the strategy profile where all SPs follow the honest strategy.
$\vec{\sigma}^{\textit{honest}}$ constitutes a Nash equilibrium.
\end{restatable}

\paragraph{(2) Mutual Dishonesty is Not an Equilibrium.}
Systems that rely solely on internal audits suffer from a failure mode where all SPs report universal success without performing actual verification, and potentially without storing any data.  
In these systems, ``mutual dishonesty" often forms a stable equilibrium.
In \name{}, the audit-the-auditor mechanism prevents this failure mode by introducing a risk of on-chain slashing for falsely reported audits.

\begin{restatable}{theorem}{mutualDishonestyNotNash}[Mutual Dishonesty is Not a Nash Equilibrium]
\label{thm:mutual-dishonesty-not-nash}
Suppose every SP follows the mutual dishonesty strategy: storing nothing and reporting success (`1') for all audits without requiring proofs from the colluding peers.  
This strategy profile is not a Nash equilibrium.
\end{restatable}

\paragraph{(3) The Honest Equilibrium is Coalition Resistant.}
Beyond individual incentive compatibility, \name{}'s auditing protocol also exhibits the desired property of robustness against coalitions.  
Specifically, no coalition of up to $f$ SPs can significantly improve their total utility by jointly deviating from honest behavior.

\begin{restatable}{theorem}{coalitionResistance}[$\varepsilon$-Coalition Resistance]
\label{thm:coalition-resistance}
Let $C \subseteq \{1, \dotsc, n\}$ be a coalition of SPs with $|C| \leq f$.
Suppose all SPs outside $C$ follow the honest strategy $\sigma^{\textit{honest}}$.
Then, under the protocol’s incentive model, any joint deviation by $C$ improves the coalition’s aggregate expected utility by at most an $\epsilon > 0$ margin, where $\epsilon$ is negligible relative to standard reward and penalty scales.
\end{restatable}

\section{Economic Opportunity}
\label{sec:Economics}

This section details the core economic factors underpinning \name{}'s storage and retrieval model. The main cost components of providing decentralized storage are physical storage, bandwidth, durability maintenance, and system coordination. These factors are carefully managed to ensure that \name{} delivers read and write performance suited to high-demand use cases, while preserving key Web3 properties such as user ownership and control. The cost structures guide the setting of system parameters - including storage rewards, read pricing, and slashing penalties - to ensure the network remains economically sustainable for both storage providers (SPs) and RPC nodes.
SPs earn rewards for storing user blobs and serving read requests; RPCs earn fees for acting as the interface between users and the network. These rewards are calibrated to reflect real-world costs, ensuring that participation is attractive for both roles.

\subsection{Write Economics}

Storing data incurs several key costs. Physical storage hardware and ingress bandwidth are fundamental contributors~\cite{AWS,GCP,backblaze,jiang2008disks}. Orchestrating the distribution of data also introduces coordination overheads, while long-term durability requires continuous repair of failed storage nodes or disks~\cite{repairreference,rashmi2013solution,Rashmi2015Cake}. The \name{} storage coding scheme (\Cref{sec:erasurecoding}) directly addresses these challenges by achieving low replication overhead and minimal repair bandwidth, without compromising fault tolerance. This is essential to ensure reliability in any distributed storage system. In a \textit{decentralized} setting, however, ensuring reliability requires an additional cost-auditing-to cryptographically verify that data remains intact and available over time. While many decentralized networks face difficult trade-offs here-often sacrificing either cost efficiency or cryptoeconomic guarantees-\name{} is engineered to maintain strong security and decentralization properties with minimal cost overhead.

The user's storage fee is $W$ per GB per month. This payment is allocated between \textit{storage rewards} ($\textit{rwd}_{\textit{st}}$ per GB) and \textit{auditor rewards} ($\textit{rwd}_{\textit{au}}$ per successful audit). To minimize payment friction, rewards are accumulated and disbursed at the end of audit epochs, ensuring that payment handling costs are negligible in aggregate. Since $\textit{rwd}_{\textit{st}}$ is a monthly rate per GB and $\textit{rwd}_{\textit{au}}$ is per audit, allocating between them requires normalization: if a GB of user data is subject to $n_a$ audits per month (on average), then the payout satisfies $\textit{rwd}_{\textit{st}} + n_a \cdot \textit{rwd}_{\textit{au}} = W$. Here, $n_a$ is derived by multiplying 3 terms (i) the expected number of chunks selected for audit in a single epoch from that a user GB ($p_a$ times number of chunks in a GB), (ii) the number of auditors per audit and (iii) the number of audit epochs per month. The parameters $\textit{rwd}_{\textit{st}}$, $\textit{rwd}_{\textit{au}}$, and $p_a$ are tuned to ensure that SPs maintain healthy profit margins while the system remains secure. Thanks to \name{}'s hybrid auditing scheme and real-world cost benchmarking (\Cref{sec:RealityBasedIncentives}), there is ample flexibility to set these parameters effectively.

Many existing decentralized storage protocols suffer from chronic underutilization: vast amounts of storage are committed and rewarded, but rarely accessed or used meaningfully. This disconnect stems from rewarding allocation rather than utility-paying for bytes pledged, not bytes served. In contrast, \name{} ties rewards to actual retention and audit performance, ensuring that incentives reflect real value. By doing so, it avoids over-incentivizing idle storage and aligns provider revenue with system usefulness.

Through careful engineering, \name{} is able to offer write pricing that remains competitive with existing Web3 storage solutions, while delivering substantially better performance. Unlike many decentralized systems that focus on archival or cold storage with minimal read optimization, \name{} is designed for high-throughput, low-latency access, making it suitable for performance-critical applications. Importantly, these performance gains are achieved without sacrificing decentralization or cryptoeconomic guarantees, and at a cost structure that remains within a practical range relative to both Web2 and Web3 alternatives.

\subsection{Read Economics}

Reading from \name{} incurs several key costs. The most significant is egress bandwidth, particularly for data served externally to end users. This is optimized by having RPC nodes decode data (even if redundant) within the dedicated network layer and egress only the minimal required data externally. While further optimization within the internal network is possible and planned, the much higher cost of external (internet-facing) egress makes internal optimizations a secondary priority.

The coding scheme introduces some computational overhead for erasure decoding during reads, but this cost is negligible in both time and money compared to bandwidth and other factors. Payment coordination in a decentralized system can also impose substantial overhead if handled naively. By employing payment channels from RPCs to SPs, \name{} effectively amortizes these costs, enabling micropayments of $\$10^{-9}$-small enough to support seamless, trust-minimized user experiences where only minimal amounts are ever at stake as payments track data delivery in real time.

An often overlooked but crucial cost is maintaining a visible catalog that tracks available data and the metadata required for retrieval. High read performance depends on this catalog being up-to-date and efficiently accessible. \name{} uses the Aptos blockchain to maintain a globally consistent catalog of verified data. However, since reads do not modify the catalog, they do not require direct blockchain interaction. Instead, RPC nodes (or any interested party) maintain a local copy of the catalog, kept current via an Aptos full node. Notably, the same catalog data structure is also used for coordinating the auditing protocol, whose costs are associated with writes-efficiently sharing catalog maintenance costs across system functions.

Reads in \name{} benefit from strong economies of scale. Major costs such as catalog maintenance remain nearly constant regardless of read volume, so the marginal profit per read increases as volume grows. As a result, calibrating the read price $R$ (per GB) is relatively straightforward, relying primarily on predictable, well-understood cost factors. High read volume benefits both the system and content owners, aligning incentives to attract and serve large-scale read traffic with high performance.

\subsection{RPC revenues}
RPC nodes play an important role in \name{} by serving as the main entry point for users, coordinating data retrieval and enhancing the overall user experience.
They require a connection to the private backbone network, an external-facing internet connection to serve user requests, and standard server-grade hardware. Importantly, most of these costs are variable and scale directly with usage-for example, bandwidth and decoding costs increase with the volume of user reads. This tight coupling of costs to revenues ensures that RPCs face minimal financial risk when scaling their operations.

The low cost structure and minimal operational risk allow RPC nodes to adopt flexible pricing models for readers, such as pay-per-use, subscription tiers, or even subsidized models in partnership with content providers. Additionally, RPCs have the opportunity to generate extra revenue by maintaining small caching layers that store pre-decoded versions of frequently accessed (``hot'') data. This not only improves read performance for popular content but also offers RPCs a clear economic incentive to invest in optimizing their service. 
To enable caching that improves overall system performance and user experience, while ensuring SPs receive their fair share of read income, \name{} implements a fee-sharing mechanism that encourages RPCs to cache ``hot data'' when it benefits the network as a whole.

\subsection{Reality Based Incentives}
\label{sec:RealityBasedIncentives}
The following are key parameters used in determining incentive bounds:

\begin{description}[style=multiline,leftmargin=2.5cm,font=\normalfont]
    \item[$\boldsymbol{c_s}$] cost to store a \spunitofstorage{} per epoch.
    \item[$\boldsymbol{c_r}$] cost to retrieve (repair) a \spunitofstorage{}.
    \item[$\boldsymbol{\textit{rwd}_{\textit{st}}}$] storage reward per \spunitofstorage{} per epoch.
    \item[$\boldsymbol{p_a}$] a stored \spunitofstorage{}'s probability of being audited per epoch.
    \item[$\boldsymbol{S_a}$] slashing penalty for audit failure.
    \item[$\boldsymbol{P_{S_a}}$] probability that a missing storage is detected via on-chain sampling.
    \item[$\boldsymbol{\textit{rwd}_{\textit{au}}}$] auditor reward per successful audit.
    \item[$\boldsymbol{p_{\textit{ata}}}$] probability of audit-the-auditor verification.
    \item[$\boldsymbol{S_{\textit{ata}}}$] slashing penalty for audit-the-auditor failure.
\end{description}

The following inequalities ensure that rational SPs behave honestly:

\paragraph{1. Incentive to participate at all.}
\[
\textit{rwd}_{\textit{st}} \ge c_s.
\]
This guarantees that the expected additional storage reward covers the expected cost of storing over time.

\paragraph{2. Incentive to store data rather than simulate it via retrieval.}
\[
p_a \cdot c_r \ge c_s.
\]
According to \Cref{lemma:storage-incentives} this ensures that a rational SP stores the data instead of retrieving on audit.

\paragraph{3. Slashing to avoid fake storage.}
To prevent an SP from faking a \texttt{prct\_fake} fraction of its total committed storage (\texttt{total\_committed}), the following should hold:
\[
P_{S_a} \cdot S_a > (1 - p_a) \cdot \textit{rwd}_{\textit{st}} \cdot \texttt{prct\_fake} \cdot \texttt{total\_committed},
\]
where $P_{S_a}$ accounts for the probability of detection through on-chain sampling. The left-hand side of the inequality has the expected penalty from faking storage while the right-hand side has the expected reward from faking the storage.
Calculating the on-chain detection probability $P_{S_a}$ needs the expected on-chain sampled audits, which are based on the SP's score.
The expected score of the SP is $(1 - \text{prct\_fake})$ and the resulting expected on-chain sample size is:
\[
(1 - (1 - \text{prct\_fake})^2) \cdot C,
\]
where $C$ is a constant set by the system (e.g. 50).
Bounding $P_{S_a}$ from below by calculating for sampling without replacement yields
\[
P_{S_a} \ge 1 - (1 - \text{prct\_fake})^{(1 - (1 - \text{prct\_fake})^2) \cdot C}.
\]
For instance, even an SP that only fakes 10\% of its storage, i.e., $\text{prct\_fake}=0.1$ and $C=50$, risks a $P_{S_a} > 0.63$ of getting caught and slashed. This is a strong deterrent.

\paragraph{4. Audit the auditor parameters.}
$\textit{rwd}_{\textit{au}}$ is set to be greater than the cost of verifying and storing the audit proof, and the amortized cost of on-chain posting the scoreboard. These are tiny and are $O(10^{-XXX})$ \$ per audit.
$S_{\textit{ata}}$ and $p_{\textit{ata}}$ are calibrated to satisfy 
\[
S_{\textit{ata}} \geq \frac{\textit{rwd}_{\textit{au}}}{p_{\textit{ata}} \cdot \varepsilon}
\]
where $\varepsilon > 0$ is a protocol-level certainty threshold controlling the confidence level required for an auditor to report a successful audit, which is defined in \Cref{def:HonestStrategy}.

\paragraph{Numerical justification.}
The basic parameters of $c_s$ and $c_r$ are mostly independent of the design, but they do influence the calibration of important design parameters. Notably, the audit frequency $p_a$ is lower bounded based on their ratio. To demonstrate the practicality of \Cref{lemma:storage-incentives} in reality, consider the most widely used storage system  -  AWS\footnote{Data taken from \url{https://aws.amazon.com/s3/pricing/} at May 8 2025. The prices are for the default presented option: S3 Standard, US East (N. Virginia). To be conservative, these are the highest storage price and lowest data transfer price.} - to provide support with real world numbers.
\begin{itemize}
    \item Storage: \$0.023/GB/month $\approx$ \$0.00000077 per day per MB.
    \item Read: \$0.02/GB $\approx$ \$0.00002 per MB.
\end{itemize}
In addition, for a $(k,m)$ erasure code, reconstructing a single \spunitofstorage{} (the specific one the SP should have but does) requires reading $k$ different \spunitofstorage{}s. In \name{}, it will be no less than $k=5$ distinct \spunitofstorage{}s to read. 
For real world prices ratio, retrieval-based strategies are disincentivized if:
\[
p_a \geq \frac{c_s}{c_r} \Leftrightarrow p_a \geq \frac{0.00000077}{5\cdot 0.00002} = \boxed{0.0076}.
\]
where we normalize $p_a$ to be the per-day probability of a chunk being audited.
That is, based on conservative pricing from popular storage services, auditing each chunk with probability at least $0.0076$ per day (i.e., once every $\sim$130 days) suffices to make retrieval-based deviations economically irrational. Since \name{} is designed to perform audits at significantly higher frequencies, \Cref{lemma:storage-incentives} is justified by rational economic behavior in practice. That is, an SP would rather pay for more storage capacity than retrieve the data from external sources.

In a similar vein to the above, we argue that rational Auditors store the audit proofs they receive. 
Roughly, the immediate cost saving from not storing the audit proof is minimal  -  only avoiding a $\sim$1~KiB sample. 
However, if the auditor is later selected for audit-the-auditor verification (with probability $p_{\textit{ata}}$), it must reproduce the corresponding audit proof. 
While the proof itself concerns only a small 1KiB sample, the system enforces chunk-level retrieval granularity: retrieving even a small part of a \spunitofstorage{} requires accessing the entire \spunitofstorage{} (e.g., \spunitofstoragesize{}).  
Thus, without further trust assumptions (e.g., trusting the auditee) the auditor would incur the full bandwidth and retrieval cost of accessing a complete \spunitofstorage{}.
Since the immediate cost saving is tiny compared to the expected retrieval expense, rational auditors are economically deterred from falsely reporting success, even when the audit-the-auditor frequency $p_{\textit{ata}}$ is very low.


\section{Forward-Looking Compute Use Cases}
\label{sec:ComputeUseCases}

According to Statista~\cite{Statista}, 149 zettabytes of data were generated globally in 2024 alone - a figure projected to more than double within five years. While most of this is media, the growth of agentic AI systems is expected to accelerate this trend even further. This exponential increase highlights a need for scalable infrastructure not just for storage, but for high-throughput data access and processing.

\name{} lays the foundation for a decentralized cloud by offering verifiable storage with predictable, low-latency access. This makes it possible to strategically position compute engines-ranging from CPUs to GPUs - within the network, enabling them to operate over large datasets without the bottlenecks typical of decentralized systems. Several forward-looking use cases stand to benefit:

\paragraph{AI and Data Marketplaces.}
Modern AI pipelines, from training to inference, are constrained by data access. \name{} supports storage of model weights, checkpoints, logs, and structured datasets for fine-tuning, retrieval-augmented generation (RAG), and similar tasks. Its read-incentivized model encourages the emergence of curated, high-quality datasets, while predictable performance allows seamless integration with inference layers and vector databases. AI data marketplaces can finally offer one-click, inference-ready access without compromising decentralization.

\paragraph{On-Chain Trading and Strategy Engines.}
As trading infrastructure moves increasingly on-chain, low-latency access to market data, user behavior, and historical execution logs becomes critical. \name{}’s throughput and read performance enable complex strategies such as real-time analytics, dynamic risk modeling, and backtesting. DeFi platforms can build time-sensitive execution logic rivaling TradFi systems-all without relying on centralized data services.

\paragraph{Edge Media Transcode.}
Different tradeoffs apply when serving media with varying levels of popularity. Highly viewed content is stored in multiple formats to minimize playback latency and compute. For less popular files, however, preemptive transcoding is inefficient. An on-demand transcoding service at RPC nodes can offer a cost-effective approach to serving the long tail of cold media content.

\subsection{Technical Paths Forward}
\label{sec:future-compute-tech}

While \name{} is designed as a decentralized storage network optimized for performance and verifiability, its architecture also opens pathways for supporting decentralized compute capabilities in the future. Several potential directions exist, each offering different trade-offs between trust, scalability, and system complexity.

\paragraph{Validator-Based Compute.}  
One straightforward option is to perform computations directly on the Aptos validator set. Since validators already maintain a Byzantine fault-tolerant (BFT) consensus, executing small tasks redundantly across all validators would inherit strong correctness guarantees.  
This approach is simple to implement and leverages existing trust assumptions.  
However, because every validator independently recomputes the task, compute amplification is high, leading to significant inefficiency.  
Validator-based compute is therefore most suitable for lightweight operations where security and simplicity outweigh performance concerns.

\paragraph{Probabilistic BFT Compute via Sampled Committees.}  
A second option is to assign computations to randomly sampled committees of storage providers or specialized compute nodes. Each committee member independently computes the result, and correctness is established via majority arguments.  
This reduces compute amplification substantially compared to full replication across validators.  
While the probabilistic security (depending on sample size and adversarial fraction assumptions) is weaker than full BFT, it remains strong for many practical workloads.  
Sampled-committee designs are increasingly common in modern decentralized systems, such as DFINITY’s Internet Computer and parts of Ethereum’s Danksharding roadmap.

\paragraph{Incentive-Driven Compute with Random Verification.}  
Another approach is to perform compute tasks off-chain and validate them selectively. Computation results would be produced by a storage provider or compute node, and a random subset of results would be verified by independent verifiers.  
If faults are detected, the provider is slashed or penalized.  
This model offers substantial efficiency gains because not every computation needs redundant recomputation.  
However, it introduces more complex economic modeling: slashing penalties must be large enough, and verification sampling rates must be high enough, to deter dishonest behavior.  
This probabilistic verification has been explored in optimistic rollups.

\paragraph{Trusted Execution Environments (TEEs).}  
Another pragmatic option is to delegate computations to hardware-based Trusted Execution Environments, such as NVIDIA Confidential Computing platforms~\cite{H100}.  
In this model, nodes would execute computations inside enclaves, producing attestations that the computation was carried out correctly.  
TEEs reduce compute overhead dramatically and simplify protocol design, since results are trusted based on hardware guarantees.  
However, TEEs introduce new trust assumptions and attack surfaces: vulnerabilities in enclave technology (e.g., side-channel attacks, speculative execution flaws) can undermine guarantees.  
Thus, while TEEs are a practical short-term option, they trade theoretical decentralization for efficiency.

\paragraph{Specialized Verifiable Computation.}  
In cases where the type of computation is known in advance and restricted (e.g., specific classes of aggregation, simple transformations, or domain-specific operations), it becomes feasible to design highly efficient verifiable computation protocols.  
Instead of general-purpose verifiable computation, these systems use custom-tailored SNARKs or succinct ZK proofs for particular computations, greatly reducing prover costs.  
Recent advances in specialized proving systems, such as zk-SNARKs for aggregation, matrix multiplications, or even certain machine learning models (ZKML), suggest that domain-specific verifiable compute could become practical much earlier than fully general approaches.  
In \name{}, enabling certain classes of predictable compute operations over stored data - and verifying them succinctly - offers an attractive middle path between theoretical purity and practical viability.

\paragraph{General Verifiable Computation via Zero-Knowledge Proofs.}  
The most theoretically appealing direction remains fully general verifiable computation, where any arbitrary program's output is accompanied by a succinct cryptographic proof (e.g., SNARKs or STARKs).  
This approach provides the strongest possible guarantees: correctness without replication, additional trust assumptions, or economic incentives.  
However, general-purpose proof generation remains prohibitively expensive today, particularly for large or complex computations.  
While promising progress is being made on zkVMs and proof systems like Halo2 and Nova, these approaches are not yet practical for general deployment at the scale envisioned for \name{}.  
Thus, while general verifiable computation remains a long-term goal, near-term efforts will likely focus on specialized or hybrid approaches.

Each path toward decentralized compute in \name{} carries distinct trade-offs between simplicity, security, efficiency, and decentralization.  
Notably, all paths are enabled and accelerated by \name{}'s high-performance data access, making integrated compute applications both feasible and profitable in decentralized settings.

\newpage

\printbibliography

\appendix

\pagebreak

\section{Durability and Availability in Web3}
\label{sec:erasurecodingappendix}

Users of data storage systems expect durability (data is never lost) and availability (data can be read right now). Unfortunately, due to 
hardware failures, perfect durability and availability is
impossible to achieve.
A simple approach where data is replicated exactly in
many places~\cite{gfs} can also achieve high durability and availability, but the
storage overhead and cost grows as a result.
Web2 storage systems achieve high data durability and availability
at reasonable cost by using erasure coding schemes like Reed-Solomon coding~\cite{reedsol}.
For example, Amazon S3 advertises 11-nines (99.999999999\%) of durability, and 99.99\% availability~\cite{S3}.
This means there is only a 0.000000001\% chance of data loss (annually), and the service is only unavailable for around 5 minutes per month.
For complete coverage of these SLAs, see~\cite{availincloud}.

The likelihood of common hardware failures are~\cite{jiang2008disks}:


\begin{itemize}
    \item \textbf{Drive Failure:} Drives fail at a rate of 2\% per year \cite{backblaze_hdd_stats} \cite{availincloud}. When a drive fails, the data stored \textit{on that specific drive} is considered lost and requires regeneration from other sources.
    \item \textbf{Latent Sector Failures:} Modern, high density hard drives experience sector failures in which a sector is unreadable due to physical issues with the drives. These errors are not typically detected until a read is attempted, and the chances of these errors occurring increase as drives age. A 2010 study \cite{lse} found that around 3.45\% of drives will experience a sector failure in their lifetime. Increasing drive density seems to make the problem worse. When a sector is found to be bad during an audit, any data which is stored by the sector is lost and requires regeneration from other sources.
    \item \textbf{Host Failure:} Hosts fail due to non-drive hardware issues, kernel panics, etc at a rate of around 1\%-5\% per year \cite{availincloud}. When a host fails, the data on its drives it not typically lost, but the time to repair a failed host may be large, so we may chose to rebuild data \textit{from all drives} in the failed host anyway to increase availability. For this reason, a host failure is often treated as a both a durability-impacting event and an availability-impacting event.
    \item \textbf{Rack Failure:} Rack-level events (network switch failure, power distribution issue, misconfiguration) cause temporary unavailability for all hosts in the rack. \name{} assumes a rate of 5\% of racks experiencing such issues per year \cite{cloudfailrate}. Data is assumed to be only inaccessible during the outage, it is \textit{not} lost.
    \item \textbf{Datacenter Failure:} Major site-wide events (large power outage, fiber cut, natural disaster preventing access) cause temporary unavailability. Data is assumed to be only inaccessible during the outage, it is \textit{not} lost. \name{} assumes around 2\% of datacenters will experience an availability issue for more than $\sim$15 minutes annually. Some highly sensitive environments, such as financial hubs and critical internet exchange points, will have multiply-redundant power and networking infrastructure to reduce the probability of failure.
    \item \textbf{Systemic Operational Event:} A large-scale unavailability event, such as critical software bugs, widespread configuration errors, major network disruptions (beyond a single DC), control plane failures, or significant operator error. Assume 1 of these happens per year and mean time to recover (MTTR) is around 30 minutes. 1 outage per year with 30 minute MTTR is an Annualized Failure Rate (AFR) of $\frac{30\text{ minutes}}{525,600\text{ minutes in year}} \approx 0.0057\%$. Data is assumed only inaccessible, not lost. Sometimes, however, outages will be much larger and extremely widespread.
\end{itemize}

A Web3 system of similar scale to a Web2 equivalent has similar challenges. 
Current Web3 storage systems trade performance for availability assuming up to
one-third of the participants are Byzantine. As a result, they require some
combination of:
\begin{itemize}
    \item \textbf{Extremely high levels of replication:} This increases cost of storage dramatically.
    \item \textbf{Sharding \userfile{}s on all nodes:} This limits scalability and performance, because network access, disk overhead, etc. dominate read performance as data shards get smaller.
    \item \textbf{Excessive communication:} This also limits scalability as nodes are added because clients exchange data with one-third of nodes for each data read.
\end{itemize}

The use of auditing and cryptoeconomic security for durability allows \name{} to avoid these performance-robbing drawbacks.

\paragraph{Durability.}
Hardware failure rate and the likelihood of malicious actors bound \name{}'s durability.
The audit scheme (see Section~\ref{sec:AuditAndIncentiveCompatibility}) finds misbehaving nodes (with some probability) and ejects them from the system.
This is a slightly different model than traditional BFT. Rather than assuming these nodes are present, \name{} assumes they are present for some duration and are removed from the storage quorum when detected.
From a durability standpoint, a malicious node is trivial (logically) to detect in  \name{}; the coordination layer asks it to prove it is storing its assigned \spunitofstorage{}s of data.
If the node does not respond or the commitments do not match, \name{} recovers \spunitofstorage{}s from other nodes to ensure durability.


The overall time a \spunitofstorage{} remains vulnerable - from when it's initially lost or corrupted until its replacement is fully rebuilt and secured - is determined by the sum of the time it takes to detect the problem and the time it takes to subsequently rebuild the data.
Let this total vulnerability period be denoted as $T_{\text{critical}} = \text{MTTD} + \text{MTTR}_{\text{rebuild}}$, where MTTD is mean time to detect.

Permanent data loss occurs if too many events cause \spunitofstorage{}s to be lost within the same \sectionuserfile{} during the critical window $T_{\text{critical}}$.

Let $\text{AFR}_{\text{effective}}$ represent the annualized rate at which these initial trigger events (plain-old node crashes or newly detected malicious actions) occur for any given node storing a \spunitofstorage{}.

As a reminder, Clay codes (Section~\ref{sec:erasurecoding}) are an MDS
scheme, meaning they can be modeled as a $(k, m)$ scheme as far as durability
is concerned.
In a $(k, m)$ scheme, $k + m$ total parts are created from source data.
According to the MDS property, any $k$ parts are sufficient to reconstruct the original
data.
Using a $(10, 6)$ scheme, one trigger event will leave 15 nodes holding \spunitofstorage{}s for that \sectionuserfile{}.
Data loss occurs if 6 of these remaining 15 nodes also experience a trigger event within $T_{\text{critical}}$.

Assuming a 12 hour MTTR, nodes have a (very high) 50\% likelihood to delete a \spunitofstorage{}, and it takes 24 hours to detect a \spunitofstorage{} was lost,
the probability of data loss is modeled as:

\begin{align*}
P(\text{data loss}) &\approx (16 \cdot 0.50) \cdot \left[ \binom{15}{6} \left( 0.50 \cdot \frac{24 + 12}{8760} \right)^6 \right] \\ 
                    &\approx 3.01 \cdot 10^{-12}
\end{align*}

Or (simplified) durability of $1 - 3.01 \cdot 10^{-12} = 0.999999999997$ (11 nines).

\paragraph{Availability.} Availability is similarly derived. 
For the same $(10, 6)$ example, data becomes temporarily unavailable for reads only if $m+1=6$ or more \spunitofstorage{}s are simultaneously down.

Availability is always strictly worse than durability. Availability is
affected by other failure modes in addition to the failure modes that affect durability.
In the derivation below, these additional failure modes are systematic error (a software issue which
can be detected and repaired quickly) and datacenter failures.

Assume 5 datacenters with 98\% uptime. Assume one systemic event per year lasting 30 minutes. Assume the EC scheme allows all \sectionuserfile{}s
to be available if any three of the datacenters are operational.

\begin{align*}
P(\text{unavail.}) &\approx P(\text{data loss}) + P(\text{systematic error}) + P(\text{$<3$ DCs online}) \\
&\approx 3.01 \cdot 10^{-12} + \frac{30\text{ min. }}{525,600\text{ min. in year }} + 1 - \left[0.98^5 + \binom{5}{4} \cdot 0.98^4 \cdot 0.02 + \binom{5}{3} \cdot 0.98^3 \cdot 0.02^2 \right] \\
&\approx 3.01 \cdot 10^{-12} + ~0.0000571 + 1 - \left[0.9039208 + 0.0922368 + 0.0037648\right]  \\
&\approx 1.35\cdot 10^{-4}
\end{align*}

Or (simplified) availability of $1 - 1.35 \cdot 10^{-4} = 0.999865$ (3 nines).

Web2 users suffer from the challenge of censorship and have no realistic recourse.
Web3 offers improved censorship resistance.
\name{} disincentivizes 
censorship by storage provider nodes in a few ways. Storage providers
do not control their ability to censor any particular data because
the placement
of \spunitofstorage{}s on storage providers is randomized by the smart contract.
Isolated nodes also do not control
user data delivery. They forgo read payments for no perceivable benefit when
censoring.



\section{Proof Sketches for Incentive Compatibility Claims}
\label{app:IC}

\LEMstorageIncentives*

\begin{proof}
If an SP deletes a chunk and relies on retrieval during an audit, its expected cost per epoch is 
$E\left[ \text{Cost}_{\text{retrieve}}\right] = p_a \cdot c_r$.
By contrast, if the SP retains the chunk, the cost is simply the storage cost $c_s$.
Thus, if $p_a \cdot c_r \ge c_s$, the expected cost of the retrieval-based strategy exceeds the cost of storing, making persistent storage strictly more profitable for rational SPs.
\end{proof}


\honestNash*

\begin{proof}[Proof sketch]
Proving the theorem is organized according to the main components of the honest strategy.\\[0.5ex]
\noindent \textit{Correct response to audits.}
Upon receiving an audit challenge, an SP must choose whether to generate and broadcast a valid audit proof. \
Deviating (e.g., ignoring the challenge or submitting a fake proof) leads to
(i) a failed/zero audit recorded by peers, decreasing the audit score and thus storage rewards, and
(ii) an increased likelihood of targeted on-chain audits.
Thus, correctly responding to audit challenges maximizes expected utility.

\noindent \textit{Persistent storage.} An SP must decide whether to persistently store its assigned \spunitofstorage{}s. \
Based on the positive utility of correctly responding to audit challenges together with \Cref{lemma:storage-incentives}, which implies that answering correctly to audits is best served by storing the data, a rational SP will choose to store its assigned \spunitofstorage{}s.

\noindent \textit{Verification and reporting of peer proofs.}
Upon receiving an audit proof $\pi_j$ from a peer, the SP must decide how to report it. \ 
Based on the calibration of $S_{\textit{ata}}$, reporting `1' yields positive expected utility only if $\pi_j$ is valid w.h.p. ($p_{\textit{inv}}(\pi_j)< \varepsilon$). Reporting slashing evidence is profitable. And, finally, $\textit{rwd}_{\textit{au}}$ is calibrated such that publishing the scoreboard is profitable (especially when dominated by `1's, as is the case here). Thus, following the reporting rule based on certainty threshold $\varepsilon$ maximizes expected utility. \\[1ex]
\indent To conclude, in every situation covered by the honest strategy - storage, responding to audits, verifying peers, and submitting slashing proofs - the prescribed honest action yields the highest expected utility given that others are also honest.
Thus, no SP has an incentive to unilaterally deviate from $\sigma_i^{\textit{honest}}$. That is, $\vec{\sigma}^{\textit{honest}}$ is a Nash equilibrium.
\end{proof}

\mutualDishonestyNotNash*

\begin{proof}[Proof sketch]
Fix SP~$i$ and consider a single reported audit entry.
Because SP~$i$ has not stored the corresponding data, it cannot produce a valid audit proof if challenged.
Each `1'-entry in the scoreboard is selected for audit-the-auditor verification with probability $p_{\textit{ata}} > 0$.
Thus, the expected utility per reported `1' is
$\textit{rwd}_{\textit{au}} - p_{\textit{ata}} \cdot S_{\textit{ata}}$,
where the protocol calibrates $S_{\textit{ata}}$ such that
$p_{\textit{ata}} \cdot S_{\textit{ata}} \gg \textit{rwd}_{\textit{au}}$.

Thus, the expected utility per reported `1' is strictly negative.
Moreover, the losses compound superlinearly with the number of false audit reports.
Therefore, SP~$i$ has a strict incentive to deviate from the mutual dishonesty strategy by not falsely reporting success.
Hence, mutual dishonesty is not a Nash equilibrium.
\end{proof}


\coalitionResistance*

\begin{proof}[Proof Sketch]
Because $f < n/3$, after trimming the highest and lowest $f$ evaluations, audit scores are bounded, from above and below, by honest SPs' reports.  
Thus, internal collusion among $C$ cannot significantly affect storage rewards or shield severe misbehavior (e.g., not storing allocated \spunitofstorage).
Moreover:
\begin{itemize}
    \item SPs in $C$ still face independent random audit challenges and audit-the-auditor checks.
    \item Failure to store assigned data or report audits correctly risks slashing and reward loss.
    \item Trust-based shortcuts (e.g., blind acceptance of peer proofs) may slightly reduce auditing costs within the coalition, but the savings are small compared to the honest strategy rewards and the penalties for misbehavior.
\end{itemize}

Thus, any aggregate gain achievable by $C$ through deviation is bounded by a small $\epsilon$ reflecting minor internal efficiency savings,
and the honest strategy profile is $\epsilon$-coalition resistant up to size $f$.
\end{proof}

\end{document}

%% file: preamble.tex
\usepackage[utf8]{inputenc}
\usepackage{amsmath, amssymb}
\usepackage{amsthm}
\usepackage{thmtools}
\usepackage{graphicx}
\usepackage{hyperref}
\usepackage{breqn}
\usepackage{enumitem}
\usepackage[table]{xcolor}
\usepackage{geometry}
\usepackage{cleveref}
\usepackage[sorting=ydnt]{biblatex} 
\usepackage{comment}
\usepackage{multirow}
\usepackage{pifont}
\newcommand{\xmark}{\ding{55}}

\usepackage{authblk}

\makeatletter
  \renewcommand\AB@authnote[1]{\textsuperscript{\normalfont\@fnsymbol{#1}}}
  \renewcommand\AB@affilnote[1]{\textsuperscript{\normalfont\@fnsymbol{#1}}}
  \renewcommand\AB@affilsepx{\hspace{2em}}
\makeatother

\usepackage{booktabs}
\usepackage{float}
\usepackage[utf8]{inputenc}
\usepackage{listings}
\usepackage[mathlines]{lineno} 

\hypersetup{
  colorlinks=true,
  linkcolor=blue,
  urlcolor=blue,
  citecolor=blue,
  pdftitle={Shelby White Paper}
}

\newtheorem{definition}{Definition}

\newcommand{\name}{Shelby}
\newcommand{\userfile}{Blob}
\newcommand{\sectionuserfile}{Chunkset}
\newcommand{\spunitofstorage}{Chunk}
\newcommand{\proofforaudit}{Sample}

\newcommand{\sectionuserfilesize}{10~MiB}
\newcommand{\spunitofstoragesize}{1~MiB}
\newcommand{\proofforauditsize}{1~KiB}
